\documentclass[11pt, letterpaper]{article}
\usepackage{fullpage}
\usepackage{amsthm}
\usepackage{amsmath,amssymb,amsfonts,nicefrac}
\usepackage{xspace}
\usepackage{color}
\usepackage{url}
\usepackage{hyperref}
\usepackage{bm}
\usepackage{bbm}
\usepackage{times}
\hypersetup{hidelinks}

\usepackage{enumitem}

\newtheorem{thm}{Theorem}[section]

\newtheorem{lemma}[thm]{Lemma}

\newtheorem{claim}[thm]{Claim}

\newtheorem{definition}[thm]{Definition}
\newtheorem{remark}[thm]{Remark}

\newcommand\card[1]{\left| {#1} \right|}
\newcommand\sett[2]{\left\{ \left. #1 \;\right\vert #2 \right\}}

\newcommand\Prob[2]{{\Pr_{#1}\left[ {#2} \right]}}

\newcommand\Expect[2]{{\mathop{\mathbb{E}}_{#1}\left[ {#2} \right]}}

\newcommand\norm[1]{\| #1 \|}

\newcommand\eps{\varepsilon}

\renewcommand\geq{\geqslant}
\renewcommand\leq{\leqslant}

\newcommand{\rom}[1]{\uppercase\expandafter{\romannumeral #1\relax}}

\title {Improved Multilayered PCPs and Hypergraph Vertex Cover\vspace{0.3cm}}
\author{ 
 \textit{Communicated by}\\[0.6em]
    \begin{tabular}{cc}
    Karthik C.S.\thanks{Department of Computer Science, Rutgers University. Supported by NSF CCF award 2313372 and NSF CAREER award 2443697}
    &
    \hspace{2em}
    Dor Minzer\thanks{Department of Mathematics, Massachusetts Institute of Technology. Supported by NSF CCF award 2227876 and NSF CAREER award 2239160.}
    \end{tabular}
    }
\date{\vspace{-5ex}}

\begin{document}
\maketitle
\begin{abstract}
 We present two elementary constructions of multilayered PCPs that improve upon prior constructions in two ways. Specifically, we give one construction of quasi-linear size, and another one with $2$-to-$2$ constraints. Using these constructions we obtain the following results for the hypergraph vertex cover problem:
 \begin{enumerate}
     \item For $k=3$, for all $\eps>0$, approximating the minimum vertex cover of a given $3$-uniform hypergraph within factor $1+\sqrt{2}-\eps$ is NP-hard. Previously, the best known result due to [Dinur, Guruswami, Khot, Regev, SICOMP 2005] achieved a factor of $2-\eps$.
     \item For $k\geq 4$, for all $\eps>0$, approximating the minimum vertex cover of a given $k$-uniform hypergraph within factor $k-\eps$ is NP-hard, which is tight. Previous works established this result  assuming the Unique-Games Conjecture [Khot, Regev, JCSS 2008], and a weaker factor of $k-1-\eps$ for standard NP-hardness [Dinur, Guruswami, Khot, Regev, SICOMP 2005].
     \item Assuming the Exponential Time Hypothesis, for all $k\geq 3$ and $\eps>0$ there is $C>0$ such that no $2^{n/\log^C n}$-time algorithm approximates the minimum vertex cover in a $k$-uniform, $n$-vertex hypergraph within factor $k-1-\eps$.
 \end{enumerate}
 The proofs were obtained using ChatGPT 5.6 Pro and subsequently rewritten by the communicators.
\end{abstract}
\section{Introduction}
To make progress on the NP hardness of hypergraph vertex cover, the work~\cite{DGKR} introduced the notion of multilayered PCPs. Multilayered PCPs have since become
an important ingredient in many subsequent PCP constructions~\cite{engebretsen2008more,KhotSaket,guruswami2014super,GL,guruswami2015inapproximability,bansal2017tight,guruswami2018strong,ABP,brandts2021complexity}. The main goal of this work is to present two multilayered PCP constructions which are stronger than the construction of~\cite{DGKR} in some ways (but weaker in other ways), and show some of their applications.

\subsection{Label Cover}
To formally state our results, we will use the label cover problem defined as follows.
\begin{definition}
An instance $\Psi$ of label cover  consists of a bipartite graph $G = (L\cup R, E)$, alphabets $\Sigma_L$, $\Sigma_R$ and a collection of constraints $\{\Phi_e \subseteq \Sigma_L\times \Sigma_R~|~e\in E\}$. 

A constraint $\Phi_e$ is called a projection constraint if there is a map $\phi_e\colon \Sigma_L\to\Sigma_R$ such that $\Phi_e = \{(\sigma,\phi_e(\sigma))~|~\sigma\in\Sigma_L\}$. If all of the constraints of $\Psi$ are projection constraints, we say that $\Psi$ is a projection label cover instance.

For $d\in\mathbb{N}$, a constraint $\Phi_e$ is called $d$-to-$1$ if it is a projection constraint and for all $\tau\in \Sigma_R$, we have that $|\phi_{e}^{-1}(\tau)|=d$. If all of the constraints of $\Psi$ are $d$-to-$1$ constraints, we say that $\Psi$ is a $d$-to-$1$ game. 

A constraint $\Phi_e$ is called $d$-to-$d$ if there are partitions $\Sigma_L = P_1\cup\ldots\cup P_m$, $\Sigma_R=Q_1\cup\ldots \cup Q_m$ such that
\begin{enumerate}
    \item $\card{P_i}=\card{Q_i}=d$ for all $i=1,\ldots,m$.
    \item $\Phi_{e} = \bigcup_{i=1}^{m}P_i\times Q_i$.
\end{enumerate}
If all of the constraints of $\Psi$ are $d$-to-$d$ constraints, we say that $\Psi$ is a $d$-to-$d$ game. 
\end{definition}
The size of $\Psi$ is defined as $|L|+|R|+|E|$, the alphabet size of $\Psi$ is defined as $|\Sigma_L|+|\Sigma_R|$, and we next define the value of $\Psi$. For an edge $e=(u,v)\in E$, we say that the labels $\sigma$ for $u$ and $\tau$ for $v$ satisfy the constraint on $e$ if $(\sigma,\tau)\in \Phi_e$. For labelings $A_{L}\colon L\to\Sigma_L$ and $A_R\colon R\to\Sigma_R$, we define 
\[
\mathsf{val}(A_L,A_R) = \frac{|\{e=(u,v)\in E~|~(A_L(u),A_R(v))\in \Phi_e\}|}{|E|},
\]
and $\mathsf{val}(\Psi) = \max_{A_L,A_R}\mathsf{val}(A_L,A_R)$. We can now define multilayered label cover.
\begin{definition}
For an integer $\ell\in\mathbb{N}$, an $\ell$-layered PCP consists of sets of vertices $X_1,\ldots,X_{\ell}$ with alphabets $\Sigma_1,\ldots,\Sigma_{\ell}$ respectively, as well as a label cover instance $\Psi_{i,j} = ((X_i\cup X_j,E_{i,j}),\Sigma_i,\Sigma_j,\Phi_{i,j})$ for each $i<j$. 
\end{definition}
We have the following two definitions of expansion for multi-layered label cover.
\begin{definition}
    Fix $\lambda\in (0,1]$, $c>0$ and let $\Psi = (\Psi_{i,j})_{1\leq i<j\leq \ell}$ be an $\ell$-layered label cover instance. 
    \begin{enumerate}
        \item \textbf{Weakly dense:} We say that $\Psi$ is $(c,\lambda)$-weakly dense if for any collection of layers $I\subseteq [\ell]$ of size at least $2/\lambda$,  and sets $A_i\subseteq X_i$ of density at least $\lambda$ for $i\in I$, there are $i<j$ in $I$ such that 
        \[
        \Prob{(u,v)\in E_{i,j}}{u\in A_i, v\in A_j}\geq c\lambda^2.
        \]
        \item \textbf{Strongly dense:} We say that $\Psi$ is $\lambda$-strongly dense if for all $i<j$, the normalized adjacency operator of the label cover instance $\Psi_{i,j}$ has second singular value at most $\lambda$. 
    \end{enumerate}
\end{definition}
It is easy to see, using the expander mixing lemma, that the property of strong density implies the property of weak density (with slightly worse parameters). Also, all known applications of multilayered PCPs use weak density. We nevertheless include the notion of ``strongly dense'' as it may be useful elsewhere.

\subsection{Main Results}\label{sec:main_results}
We first present our multilayered PCPs, and we begin with the one that is based on the $2$-to-$1$ Games theorem~\cite{KMS1,DKKMS1,DKKMS2,KMS2}. The result reads as follows:
\begin{thm}\label{thm:multilayered_222}
    For all $\lambda\in (0,1]$, $\eps>0$ and $\ell\in\mathbb{N}$ there exists $C>0$ such that the following holds. Given an $\ell$-layered, $(1/2, \lambda)$-weakly dense label cover instance $\Psi = \{\Psi_{i,j}\}_{1\leq i<j\leq \ell}$ such that each $\Psi_{i,j}$ is a $2$-to-$2$ game with the same alphabet of size at most $C$, it is NP-hard to distinguish between the following two cases:
    \begin{enumerate}
        \item \textbf{Yes case:} There are sets $X_{i}'\subseteq X_i$ of fractional size at least $1-\eps$, and partial labelings $A_i\colon X_i'\to\Sigma_i$ that satisfy all of the constraints of $\Psi$ on edges between the sets $X_{i}'$. 
        \item \textbf{No case:} For all $i<j$, $\mathsf{val}(\Psi_{i,j})\leq \eps$.
    \end{enumerate}
\end{thm}

Our second multilayered PCP is based on the quasi-linear size PCP of~\cite{BMV}. The result reads as follows:
\begin{thm}\label{thm:multilayered_expand}
    For all $\lambda\in (0,1]$, $\eps>0$ and $\ell\in\mathbb{N}$ there exists $C>0$ such that the following holds. There is a polynomial time reduction that given a $3$-SAT formula $\varphi$ of size $n$, produces an $\ell$-layered, $\lambda$-strongly dense, projection label cover instance $\Psi = \{\Psi_{i,j}\}_{1\leq i<j\leq \ell}$ of size $n\log^C n$, alphabet sizes at most $C$ such that:
    \begin{enumerate}
        \item \textbf{Yes case:} If $\varphi$ is satisfiable, then there are labelings $A_i\colon X_i\to\Sigma_i$ that satisfy all of the constraints of $\Psi$.
        \item \textbf{No case:} If $\varphi$ is unsatisfiable, then for all $i<j$, $\mathsf{val}(\Psi_{i,j})\leq \eps$.
    \end{enumerate}
\end{thm}

\subsubsection{Implications for Hypergraph Vertex Cover}
The most interesting implications of our multilayered PCPs are to the hypergraph vertex cover problem. Given a $k$-uniform hypergraph $\mathcal{H} = (\mathcal{V},\mathcal{E})$, the goal in the vertex cover problem is to find the smallest set of vertices $S\subseteq \mathcal{V}$ that touches all edges of $\mathcal{H}$. A simple polynomial time approximation algorithm achieves approximation ratio $k$, and it is known to be optimal assuming the Unique-Games Conjecture~\cite{KhotRegev}. 
For $k=2$, the best known NP-hardness result, due to~\cite{KMS1,DKKMS1,DKKMS2,KMS2}, currently stands at factor $\sqrt{2}-\eps$ for all $\eps>0$. For $k\geq 3$, the best known NP-hardness result, due to~\cite{DGKR}, stands at factor $k-1-\eps$ for all $\eps>0$. The result below improves the NP-hardness factor for all $k\geq 3$, and is optimal for all $k\geq 4$.
\begin{thm}\label{thm:hypergraph_VC_improved_factor}
    Let $k\geq 3$ be an integer.
    \begin{enumerate}
        \item For $k=3$, for all $\eps>0$, given a $3$-uniform weighted hypergraph $\mathcal{H}$, it is NP-hard to approximate the minimum weight of a vertex cover in $\mathcal{H}$ within factor $1+\sqrt{2}-\eps$.
        \item For $k\geq 4$, for all $\eps>0$, given a $k$-uniform weighted hypergraph $\mathcal{H}$, it is NP-hard to approximate the minimum weight of a vertex cover in $\mathcal{H}$ within factor $k-\eps$.
    \end{enumerate}
    \end{thm} 
    Next, we show a result that matches the ratio in~\cite{DGKR}, but implies a near exponential time lower bound assuming ETH. Its proof is the same as the construction in~\cite[Section 5]{DGKR}, except that one uses the multilayered PCP from Theorem~\ref{thm:multilayered_expand} instead of theirs, and we omit the details.
    \begin{thm}\label{thm:hypergraph_ETH}
        Assuming ETH, for all $k\geq 3$ and $\eps>0$, there exists $C>0$ such that no $2^{n/\log^{C} n}$-time algorithm approximates the minimum weight of a vertex cover of a given $k$-uniform, $n$-vertex hypergraph within factor $k-1-\eps$.
    \end{thm}

    Given Theorem~\ref{thm:hypergraph_VC_improved_factor} and Theorem~\ref{thm:hypergraph_ETH}, it is natural to ask if there are sub-exponential time algorithms achieving better than $k-1$ approximation for vertex cover over $k$-uniform hypergraphs. The only avenue we have for such hardness results relies on $d$-to-$1$ games, which inherently incurs polynomial blow-up in the reduction~\cite{ABS,Steurer2010Dto1}. Thus, it seems unlikely that one could improve upon the factor in Theorem~\ref{thm:hypergraph_ETH} using current techniques.

    \begin{remark}
     We remark that Theorems~\ref{thm:hypergraph_VC_improved_factor} and~\ref{thm:hypergraph_ETH} can be converted into a hardness result for standard, unweighted hypergraphs by vertex duplications. Moreover, the results hold for bounded degree hypergraphs.
    \end{remark}
    
    \subsection{Other Implications}\label{sec:other_implications}
    The results from Section~\ref{sec:main_results} lead to  improved hardness of approximation results, and we mention a few of them below. 
    \paragraph{Simple hypergraphs:} A $k$-uniform hypergraph is called simple if any two distinct hyperedges share at most a single vertex. Using the result of~\cite{GS}, we get the following version of Theorem~\ref{thm:hypergraph_VC_improved_factor} for simple hypergraphs.
    \begin{thm}\label{thm:simple_hyp}
        Assume $\mathsf{NP}\not\subseteq\mathsf{BPP}$. Then, for all $\eps>0$, no polynomial time algorithm approximates the size of the smallest vertex cover in a given simple $3$-uniform hypergraph within factor $1+\sqrt{2}-\eps$. For $k\geq 4$, no polynomial time algorithm approximates the size of the smallest vertex cover in a given simple $k$-uniform  hypergraph within factor $k-\eps$.
    \end{thm}

     \paragraph{Feedback vertex sets for bounded cycles:}  The input is an $n$-vertex graph, and the goal is to find as small a set of vertices as possible that intersects each cycle of size at most $O\left(\frac{\log n}{\log\log n}\right)$. The work~\cite{GL} showed a reduction from the hypergraph vertex cover problem to the minimum feedback set for bounded cycles problem. Using their reduction with Theorem~\ref{thm:hypergraph_VC_improved_factor} gives the following result:
     \begin{thm}\label{thm:feedback}
         Let $\Gamma_3 = \sqrt{2}+1$, $\Gamma_{k}=k$ for $k\geq 4$ and assume that $\mathsf{NP}\not\subseteq\mathsf{BPP}$. Then for all $k\geq 3$, $\eps>0$, there is $c_{k,\eps}>0$ such that there is no polynomial time algorithm that given an $n$-vertex graph $G=(V,E)$ distinguishes between the following two cases:
         \begin{enumerate}
             \item \textbf{Yes case:} There is a set $S\subseteq V$ of fractional size at most $\frac{1}{\Gamma_k}+\eps$ that intersects all cycles of length at most $c_{k,\eps}\frac{\log n}{\log\log n}$ in $G$.
             \item \textbf{No case:} For all $S\subseteq V$ of fractional size $1-\eps$, $G$ has a cycle of length $k$ that does not intersect $S$.
         \end{enumerate}
     \end{thm}
     Previously, the work~\cite{GL} established a version of Theorem~\ref{thm:feedback} under the Unique-Games Conjecture, and gave the weaker value of $\Gamma_k = k-1$ under the more standard $\mathsf{NP}\not\subseteq\mathsf{BPP}$.
     \paragraph{$H$-transversal:} Given a graph $H$ on $k$ vertices (thought of as small), an input to the $H$-transversal problem is a graph $G=(V,E)$ and the goal is the smallest $S\subseteq V$ such that the subgraph induced by $V\setminus S$ does not have $H$ as a subgraph. The work~\cite{GL} showed a reduction from the hypergraph vertex cover problem to the $H$-transversal problem for any $2$-connected graph, i.e., any graph that cannot become disconnected by the removal of a single vertex. Using their reduction with Theorem~\ref{thm:hypergraph_VC_improved_factor} gives the following result:
     \begin{thm}\label{thm:transversal}
         Suppose $\mathsf{NP}\not\subseteq\mathsf{BPP}$. Then, for all $\eps>0$ and for fixed $2$-connected graph $H$ with $k\geq 4$ vertices, no polynomial time algorithm approximates $H$-transversal within factor $k-\eps$.
     \end{thm}
     Previously, the work~\cite{GL} established a version of Theorem~\ref{thm:transversal} under the Unique-Games Conjecture, and the weaker factor of $k-1-\eps$ under the more standard $\mathsf{NP}\not\subseteq\mathsf{BPP}$. A similar improvement follows for the result of~\cite{LWTree}.
\subsection{Techniques}\label{sec:techniques}
\subsubsection{Multilayered PCPs}
The proofs of Theorems~\ref{thm:multilayered_222} and~\ref{thm:multilayered_expand} rely on the basic observation that the \emph{collision game} of a given projection label cover instance can be naturally viewed as a multilayered PCP. Given  a projection label cover instance $\Psi = ((L\cup R, E), \Sigma_L, \Sigma_R, \{\phi_e\}_{e\in E})$, its collision game is the game on the bipartite graph whose sides are copies of $L$, and its edges are weighted according to the following process:
\begin{enumerate}
    \item Sample $v\in R$ with probability proportional to its degree.
    \item Sample $u,u'$ neighbours of $v$ independently.
    \item Output the edge $(u,u')$.
\end{enumerate}
In principle, an edge $(u,u')$ may be formed by more than a single choice of $v$, and these can be thought of as parallel constraints. The alphabet of each side is $\Sigma_L$, and the constraint on the edge $(u,u')$ is that $\phi_{u,v}(\sigma) = \phi_{u',v}(\sigma')$. In words, the projections of their labels agree on $v$. It is well known that if $\Psi$ is satisfiable, then its collision game is satisfiable, and if $\Psi$ has small value, then its collision game has low value. Our key observation is that one can construct a multilayered PCP by taking $\ell$ copies of $L$, and having between any two of them the collision game of $\Psi$. This immediately gives Theorem~\ref{thm:multilayered_222} from the result of~\cite{KMS1,DKKMS1,DKKMS2,KMS2}.

Theorem~\ref{thm:multilayered_expand} requires a little more work due to two reasons: first, by default the collision game only satisfies the ``weakly dense'' notion, and second, it is not a projection game.
\begin{enumerate}
    \item To address the first issue, we show a black-box reduction that maps an arbitrary constant degree label cover instance (which is known to be hard), to label-cover instances over bipartite expander graphs using the transformation from~\cite{MoshkovitzFortification}. It takes a bipartite expander graph $H = (A\cup B, E)$ of constant degree whose right side is identified with $L$, and then thinks of a vertex $a\in A$ as holding the labels of all of its neighbours in $L$. The edges and constraints naturally correspond to walks from $A$ to $R$ of length $2$.  
\item To address the second issue, we show that if the original label cover has constant degrees, then one can enlarge the alphabets of the resulting collision game and make it into a projection game. This enlargement corresponds to a vertex $v$ holding the labels of all vertices $u$ of certain distance from it in the collision graph, and is similar to the way that the works~\cite{MoshkovitzRaz,DinurHarsha} ``flip'' the direction of projection constraints.
\end{enumerate}

\subsubsection{Hypergraph Vertex Cover}
The proof of Theorem~\ref{thm:hypergraph_VC_improved_factor} uses the $p$-biased long-code machinery of~\cite{DinurSafra,KhotRegev}. Its vertex set is $P([n])$ weighted by the probability distribution $\mu_p(A) = p^{|A|}(1-p)^{n-|A|}$, where the intention is to encode an element $\sigma \in [n]$ (which is a label in some outer label cover instance) via the set of vertices $\{A~|~A\ni\sigma\}$, called the dictatorship of $\sigma$. The edges/hyperedges of the $p$-biased long code are designed with the goal of ensuring that any large independent set corresponds to a few elements, and to check constraints of the outer label cover instance. Clearly, the weight of a dictatorship set is $p$, and to get the best hardness result for vertex cover one wishes to take as large as possible $p$. The main question becomes whether, for a given $p$, one can design hyperedges that ensure that:
\begin{enumerate}
    \item \textbf{Decoding:} Any large independent set can be ``decoded'' into a small set of elements.
    \item \textbf{Constraint checking:} For any constraint $(x,y)$, the small sets of ``decoded elements'' of $x$ and of $y$, contain a pair of labels satisfying the constraint on $(x,y)$.
\end{enumerate}
In~\cite{DinurSafra,KhotRegev} and here, the first of these items is handled by strong results in Boolean functions analysis such as the Friedgut junta theorem and the Russo-Margulis lemma. We thus focus on the second item.

In our case (as well as in the case of~\cite{DinurSafra,DGKR}), since we do not rely on the Unique-Games conjecture, we have to consider $2$ separate long-codes that are linked by a $2$-to-$2$ constraint. Specifically, suppose that $\Sigma_1,\Sigma_2$ are alphabets of the same even size (they can be thought of as the same, but we distinguish them for notational clarity), suppose $\Phi\subseteq \Sigma_1\times \Sigma_2$ is a $2$-to-$2$ constraint, and consider the $p$-biased long-codes $P(\Sigma_1)$, $P(\Sigma_2)$ with $p$ to be determined. We wish to add $k$-uniform hyperedges between these two long codes so as to ensure that
\begin{enumerate}
    \item For any $(\sigma_1,\sigma_2)\in \Phi$, the dictatorships defined by $\sigma_1$, $\sigma_2$ have no $k$-uniform hyperedges.
    \item For any $(\sigma_1,\sigma_2)\not\in \Phi$, the dictatorships defined by $\sigma_1$, $\sigma_2$ do have $k$-uniform hyperedges.
    \item These two properties are robust in the sense that they hold even for families that are close to the aforementioned dictatorships.
\end{enumerate}
Write $k=a+b$ where $a,b\geq 1$ are integers, and consider $(A_1,\ldots,A_a,B_1,\ldots,B_b)$ where the $A$'s come from $P(\Sigma_1)$ and the $B$'s come from $P(\Sigma_2)$. If we wish to include this as a hyperedge, for the first property to hold, we must have that 
$A_1\cap \ldots \cap A_{a}$, $B_1\cap \ldots\cap B_b$ include no pair from $\Phi$. We can afford to include all such hyperedges, and in fact it is easily seen that these hyperedges also achieve the second item. This approach, with $a=1$ and $b=k-1$, was taken by~\cite{DGKR}, who showed a combinatorial analysis that avoided the third item, but only allows $p=1-\frac{1}{k-1}$ (which results in hardness for factors up to $k-1$).

Explicitly addressing the third item amounts to designing a probability distribution over hyperedges $(A_1,\ldots,A_a,B_1,\ldots,B_b)$, in which the marginal distribution of each vertex is the $p$-biased over its respective long code. For example, the work~\cite{KhotRegev} establishes a very simple such coupling for $1$-to-$1$ constraints, where they take $a=1$, $b=k-1$ and $p=1-\frac{1}{k}$. The work~\cite{DinurSafra} established a more elaborate coupling for $1.5$-to-$1.5$ constraints for $k=2$, $a=1$, $b=1$ for $p=(3-\sqrt{5})/2$, and the work~\cite{KMS1} used a similar coupling in the context of $2$-to-$2$ games to get $p=1-1/\sqrt{2}$.

\vspace{-2ex}
\paragraph{Our improvement:} We show that if $\Phi$ is a $2$-to-$2$ constraint and $k\geq 3$, then there is a way to construct couplings for values of $p$ larger than $1-\frac{1}{k-1}$, thereby improving upon~\cite{DGKR}. 
To get some intuition, write $\Phi = \bigcup_{i=1}^{m} P_i\times Q_i$ where $\{P_i\}_{i=1,\ldots,m}$ is a partition of $\Sigma_1$ into sets of size $2$, and $\{Q_i\}_{i=1,\ldots,m}$ is a partition of $\Sigma_2$ into sets of size $2$, and write $P_i=\{x_i,x_i'\}$ and $Q_i = \{y_i,y_i'\}$. Let us attempt to construct a coupling for the simple case that $a=1$ and $p=1-\frac{1}{k-1}$. In that case, for $A_1$ to have the correct marginal, for each $i$ we need to include at least one of the elements from $P_i$ in $A_1$ with probability $q = p^2+2p(1-p)$, and in that case we must arrange that no element from $Q_i$ appears in $B_1\cap\ldots\cap B_{k-1}$. This can be achieved in one of two ways:
\begin{enumerate}
    \item Choosing distinct $j,j'$, then omitting $y_i$ from $B_{j}$ and $y_i'$ from $B_{j'}$, and including both in the rest of the $B$'s.
    \item Choosing some $j$, then omitting both $y_{i},y_{i}'$ from $B_j$ and including both in the rest of the $B$'s.
\end{enumerate} 
Since the probability $B_s\cap Q_i=\empty$ has to be $(1-p)^2$, we can afford to take the second option with probability $\frac{(1-p)^2}{q}$, and in the rest of the probability we choose the first item.  In the case $A_1\cap P_i=\emptyset$, we can afford to include both $y_i,y_i'$ in all of the $B_j$.

A direct calculation shows that the probability that $B_s\cap Q_i = Q_i$ is equal to, 
\begin{equation}\label{eq:coupling_intui}
(1-p)^2+q\left(\frac{(1-p)^2}{q}\frac{k-2}{k-1}+\left(1-\frac{(1-p)^2}{q}\right)\frac{k-3}{k-1}\right)
= 1-\frac{2}{k-1}+\frac{3}{(k-1)^3}.
\end{equation} 
On the other hand, $p^2 = 1-\frac{2}{k-1}+\frac{1}{(k-1)^2}$, and direct comparison shows that~\eqref{eq:coupling_intui} exceeds $p^2$ for $k=3$. This means that we overshot the probability of $B_s\cap Q_i = Q_i$ (but got the probability that $B_s\cap Q_i = \emptyset$ correctly), so by omitting more elements from the $B_j$ sometimes we can arrange that its marginal is $p$-biased. In fact, for $k=3$, this slack implies that such a coupling can be constructed for some $p > 1-\frac{1}{k-1}$, already giving an improvement over~\cite{DGKR}.

The above coupling shows the potential gain in having the ability to trade off omissions between several $B$'s, but concretely it only works for $k=3$, and for a modest gain. Intuition suggests that such a tradeoff for both the $A$'s and $B$'s may potentially lead to couplings for higher values of $p$ (but this is only possible for $k\geq 4$). This turns out to be correct, and when $a,b$ are both at least $2$ there are already couplings that go all the way to $p=1-\frac{1}{k}$.\footnote{This is the limit, as in expectation, for each $i$, the number of elements from $P_i$ not in $A_1\cap\ldots\cap A_a$ plus the number of elements from $Q_i$ not in $B_1\cap\ldots\cap B_b$ is at most $2k\cdot(1-p)<2$ for $p>1-\frac{1}{k}$.}

\paragraph{Statement of AI Use:} The communicators originally obtained the results presented here through interactive sessions with ChatGPT 5.6 Pro. In a separate follow-up interaction, the model independently derived the results in a single attempt. The present manuscript was subsequently written by the communicators to make the results accessible to a broad TCS audience.

\section{Preliminaries}
\paragraph{Notations:} We denote $[n] = \{1,\ldots,n\}$. For a graph $G = (V,E)$ and a vertex $v\in V$, we denote by $N_{\leq r}(v)$ the set of vertices reachable from $v$ in at most $r$ steps. For sets $A,B$, we denote by $A\Delta B$ the symmetric difference of $A$ and $B$.
\subsection{PCPs}
We will use the $2$-to-$1$ Games with imperfect completeness~\cite{KMS1,DKKMS1,DKKMS2,KMS2}, and we make two remarks. First, the completeness in the version stated below is not the same as in the papers, but their construction can easily be seen to achieve it (and is crucial towards their application to vertex cover). Second, their proof gives weighted instances, but using standard techniques (e.g.~\cite{Trevisan,MoshkovitzRaz,DinurHarsha}) one can get unweighted, biregular instances.
\begin{thm}\label{thm:221}
    For all $\eps>0$ there exists $C\in\mathbb{N}$ such that given a biregular $2$-to-$1$ game $\Psi = ((L\cup R,E),\Sigma_L,\Sigma_R,\{\phi_{e}\}_{e\in E})$ with alphabet size at most $C$ and degree at most $C$, it is NP-hard to distinguish between the following two cases:
    \begin{enumerate}
        \item \textbf{Yes case:} There exist $L'\subseteq L$ of fractional size at least $1-\eps$, a partial labeling $A_{L'}\colon L'\to\Sigma_{L}$ and a labeling $A_{R}\colon R\to \Sigma_R$ that satisfy all of the constraints between $L'$ and $R$.
        \item \textbf{No case:} No labelings $A_{L}\colon L\to\Sigma_L$, $A_R\colon R\to\Sigma_R$ satisfy more than an $\eps$ fraction of the constraints.
    \end{enumerate}
\end{thm}

We will use the quasi-linear size PCP construction of~\cite{BMV}. As before, their proof gives weighted instances, but using standard techniques (e.g.~\cite{Trevisan,MoshkovitzRaz,DinurHarsha}) one can get unweighted, biregular instances of constant degree.
\begin{thm}\label{thm:quasilinear}
    For all $\eps>0$ there exists $C\in\mathbb{N}$ such that the following holds. There is a polynomial time reduction that given a $3$-SAT formula $\varphi$ of size $n$, produces a biregular projection label cover instance $\Psi = ((L\cup R,E),\Sigma_L,\Sigma_R,\{\phi_{e}\}_{e\in E})$ of size at most $n\log^C n$,  alphabet of size at most $C$ and degree at most $C$, such that the following holds:
    \begin{enumerate}
        \item \textbf{Yes case:} If $\varphi$ is satisfiable, then $\mathsf{val}(\Psi)=1$.
        \item \textbf{No case:} If $\varphi$ is unsatisfiable, then $\mathsf{val}(\Psi)\leq \eps$.
    \end{enumerate}
\end{thm}

\subsection{Biased Long-code}
We will use the $p$-biased long-code and associated analytical machinery from~\cite{DinurSafra}. We begin with the definition of the $p$-biased Kneser graph.
\begin{definition}
    For $p\in (0,1)$ and a finite set $\Sigma$, the weighted Kneser graph $G_{p}(\Sigma)$ has the vertex set $P(\Sigma)$, and $(A,B)$ is an edge if $A\cap B=\emptyset$. The weight of the vertices is defined by $\mu_p(A) = p^{|A|}(1-p)^{|\Sigma\setminus A|}$.
\end{definition}
Next, we define the measure and the influence of a family of vertices in the Kneser graph.
\begin{definition}
    Let $\mathcal{F}\subseteq P(\Sigma)$.
    \begin{enumerate}
        \item We define $\mu_p(\mathcal{F}) = \Prob{A\sim \mu_p}{A\in \mathcal{F}}$.
        \item For $\sigma\in \Sigma$, we define the influence of $\sigma$ by $I_{\sigma}[\mathcal{F}; \mu_p] = \Prob{A\sim\mu_p}{\text{exactly one of $A$,$A\Delta\{\sigma\}$ is in $\mathcal{F}$}}$, and define the total influence of $\mathcal{F}$ by 
        $I[\mathcal{F}; \mu_p] = \sum\limits_{\sigma\in\Sigma}I_{\sigma}[\mathcal{F}; \mu_p]$.
    \end{enumerate}
\end{definition}
We also define the notion of monotone families.
\begin{definition}
    We say a family $\mathcal{F}\subseteq P(\Sigma)$ is monotone if $A\in\mathcal{F}$ and $A\subseteq B$, then $B\in\mathcal{F}$.
\end{definition}

We will also make use of the Russo-Margulis Lemma~\cite{Russo,Margulis}, stated below.
\begin{lemma}\label{lem:russo_margulis}
Suppose $\mathcal{F}\subseteq P(\Sigma)$ is monotone. Then
$\frac{d\mu_p(\mathcal{F})}{dp}(p) = I[\mathcal{F};\mu_p]$.
\end{lemma}

We will also need Friedgut's junta theorem~\cite{Friedgut}. To state it, we first define a junta.
\begin{definition}
    Let $J\subseteq \Sigma$, $\mathcal{F}\subseteq P(\Sigma)$. We say that $\mathcal{F}$ is a $J$-junta if there exists $\mathcal{J}\subseteq P(J)$ such that
    \[
    \mathcal{F} = \sett{A\subseteq \Sigma}{A\cap J\in \mathcal{J}}.
    \]
\end{definition}
In words, if $\mathcal{F}$ is a $J$-junta, then membership in $\mathcal{F}$ only depends on the projection of $A\subseteq \Sigma$ on $J$. 
\begin{lemma}\label{lem:friedgut}
    Suppose $\mathcal{F}\subseteq P(\Sigma)$ satisfies that $I[\mathcal{F};\mu_p]\leq K$. Then for every $\eps>0$ there exists $J\subseteq \Sigma$ of size $2^{O_{p,\eps}(K)}$ and a $J$-junta $\mathcal{J}\subseteq P(\Sigma)$ such that $\mu_p(\mathcal{F}\Delta\mathcal{J})\leq \eps$.
\end{lemma}

\paragraph{Restrictions:} Given a family $\mathcal{F}\subseteq P(\Sigma)$, $I\subseteq \Sigma$ and $A\subseteq I$, we define the restriction 
\[
\mathcal{F}_{I\rightarrow A}
=\sett{S\subseteq \Sigma\setminus I}{S\cup A\in\mathcal{F}}.
\]
We will use the fact that on the $p$-biased measure, restricting $K$ coordinates increases the individual influence of each other coordinate by factor at most $(\min(p,1-p))^{-K}$. 

\section{The Multilayered PCPs}
In this section we show our multilayered PCP constructions.
\subsection{Proof of Theorem~\ref{thm:multilayered_222}}\label{sec:222_multi} 
Start with a label cover instance $\Psi=((L\cup R,E),\Sigma_L,\Sigma_R,\{\phi_{e}\}_{e\in E})$ as in Theorem~\ref{thm:221} with $\eps$. The idea is to produce a multilayered version of the basic collision game of $\Psi$ as discussed in Section~\ref{sec:techniques}. More precisely, we produce the instance $\Psi'$ that has $\ell$ layers, layer $i$ is $L\times\{i\}$ and has alphabet $\Sigma_L$. For $1\leq i<j\leq \ell$, the constraints in $\Psi'_{i,j}$ are defined and weighted as follows:
\begin{enumerate}
    \item Sample $v\in R$ uniformly.
    \item Sample $u,u'\in L$ neighbours of $v$ independently.
    \item Output the edge $((u,i),(u',j))$ with the constraint 
    \[
    \Phi_{(u,i),(u',j)}=\{(\sigma,\sigma')~|~\exists \tau\in \Sigma_R\text{ such that }\phi_{u,v}(\sigma) = \phi_{u',v}(\sigma')=\tau\}.
    \]
\end{enumerate}
This completes the description of the reduction. The completeness and the fact that each $\Psi_{i,j}'$ is a $2$-to-$2$ game are obvious. As per the soundness, if there are assignments $A_{i}\colon L\times \{i\}\to \Sigma_L$, 
$A_{j}\colon L\times \{j\}\to \Sigma_L$ satisfying at least $\eps$ weight of the constraints of $\Psi'_{i,j}$, define the randomized assignment $A_R\colon R\to \Sigma_R$ where for each $v\in R$, we sample a neighbour $u\in L$ of $v$ uniformly, and then define
\[
A_R(v) = \phi_{u,v}(A_j((u,j))).
\]
Fix an edge $(u,v)$ of $\Psi$, and let $u'$ be the neighbour that $v$ picked. Note that if $(u,u')$ is satisfied in $\Psi_{i,j}'$, then $(u,v)$ would be satisfied by $u\rightarrow A_i((u,i))$, $v\rightarrow A_R(v)$. It follows that the expected fraction of edges that are satisfied by $A_i$ (thinking of it as an assignment to $L$) and $A_R$ is equal to the weight of constraints satisfied by $A_i$, $A_j$, and thus at least $\eps$.

As per the weakly-dense property, fix $I$ of size at least $2/\lambda$, and $A_i\subseteq L\times\{i\}$ of density exactly $\lambda$ for $i\in I$. Sample $v\in R$ uniformly, $u,u'\in L$ neighbours of $v$, and let $Z_i = 1_{(u,i)\in A_i}$, 
$Z_{i}' = 1_{(u',i)\in A_i}$. Then
\[
\Expect{i,j\in I}{\Expect{v,u,u'}{Z_iZ_j'}}
=
\Expect{v}{\Expect{i,u}{Z_i}^2}
\geq 
\Expect{v}{\Expect{i,u}{Z_i}}^2
\geq \lambda^2.
\]
The contribution from the case that $i=j$ is at most $\lambda^2/2$, so we conclude there are $i<j$
such that $\Expect{v,u,u'}{Z_iZ_j'}\geq \frac{1}{2}\lambda^2$, as required.

\subsection{A Quasi-linear, Expanding Multilayered PCP: Proof of Theorem~\ref{thm:multilayered_expand}}
We isolate one of the ideas in the proof of Theorem~\ref{thm:multilayered_expand} by showing the proof in the case $\ell=2$. More precisely, we show the following result:
\begin{thm}\label{thm:quasilinear_expanding_2}
    For all $\lambda\in (0,1]$, $\eps>0$ there exists $C\in\mathbb{N}$ such that the following holds. There is a polynomial time reduction that given a $3$-SAT formula $\varphi$ of size $n$, produces a biregular, projection label cover instance $\Psi = ((L\cup R,E),\Sigma_L,\Sigma_R,\{\phi_{e}\}_{e\in E})$ of size at most $n\log^C n$, alphabet of size at most $C$ and degree at most $C$, whose normalized adjacency operator has second singular value at most $\lambda$, such that the following holds:
    \begin{enumerate}
        \item \textbf{Yes case:} If $\varphi$ is satisfiable, then $\mathsf{val}(\Psi)=1$.
        \item \textbf{No case:} If $\varphi$ is unsatisfiable, then $\mathsf{val}(\Psi)\leq \eps$.
    \end{enumerate}
\end{thm}
\begin{proof} 
    Start with a label cover instance $\Psi=((L\cup R,E),\Sigma_L,\Sigma_R,\{\phi_{e}\}_{e\in E})$  as in Theorem~\ref{thm:quasilinear}.
    Denote $N = |L|$ and take a bipartite biregular graph $H = (V\cup L, E')$ with sides of size $N$ where we identify its right side with $L$, degree $D = O_{\lambda}(1)$ and second singular value at most $\lambda$; such a graph can be found using the result of~\cite{Alon2021Expanders}. We define the label cover instance $\Psi' = (V\cup R, E'', \Sigma_L^{D}, \Sigma_R,\{\phi_{e}'\}_{e\in E''})$ as follows. For $x\in V$, fixing an ordering of its neighbours $u_1,\ldots,u_D$, in $H$  a label for $x$ is $(\sigma_1,\ldots,\sigma_D)\in\Sigma_L^D$, thought of as an assignment to $u_1,\ldots,u_D$. The edges of $\Psi'$ correspond to walks of length $2$, where the first step is according to $H$, and the second step is according to $(L\cup R,E)$. Letting this walk be denoted by $(x,u,v)$, and letting $i$ be the position of $u$ as a neighbour of $x$, the constraint  is defined as $\phi_{(x,u,v)}'(\sigma_1,\ldots,\sigma_D) = \phi_{(u,v)}(\sigma_i)$. This completes the construction. The size, alphabet size and the completeness are clear.

    For the soundness, suppose that $A_V, A_R$ are assignments to $\Psi'$, and consider the randomized assignment $A_L$ as follows: for a vertex $u$, sample a neighbour $x$ of $u$ in $H$, denote by $i$ the position of $u$ as a neighbour of $x$, and define $A_L(u) = A_V(x)_i$. Then fixing $v$ and sampling a neighbour $u$ of $v$ uniformly in $\Psi$, we see that the probability that the constraint $(u,v)$ is satisfied is equal to
    \[
    \Prob{u, (x,v)\in E''}{A_V, A_R\text{ satisfy the constraint $(x,v)$ formed by $u$}}.
    \]
    Taking expectation over $v$, we conclude that the expected fraction of constraints satisfied by $A_L,A_R$ is equal to the expected fraction of constraints satisfied by $A_V,A_R$. In particular, $\mathsf{val}(\Psi')\leq \mathsf{val}(\Psi)$.

    Finally, for the expansion, denoting the normalized adjacency operators of $(L\cup R,E)$ and $H$ by $\mathrm{T}_{\Psi}\colon L_2(L)\to L_2(R)$, $\mathrm{T}_H\colon L_2(V)\to L_2(L)$ respectively (from left to right), the normalized adjacency operator of $\Psi'$ is 
    $\mathrm{T}_{\Psi}\mathrm{T}_H\colon L_2(V)\to L_2(R)$. Thus, for any $f\colon V\to \mathbb{R}$ with average $0$ we have
    \[
    \norm{\mathrm{T}_{\Psi}\mathrm{T}_Hf}_2
    \leq \norm{\mathrm{T}_H f}_2
    \leq \lambda\norm{f}_2,
    \]
    where we use that $\mathrm{T}_{\Psi}$ is an averaging operator and thus a contraction, and that $\mathrm{T}_{H}$ has second singular value at most $\lambda$.
\end{proof}
With the construction of Theorem~\ref{thm:quasilinear_expanding_2} in hand, the proof of Theorem~\ref{thm:multilayered_expand} proceeds by designing a multilayered collision game as in Section~\ref{sec:222_multi}, and then regaining the projection property by enlarging the alphabets. 
\begin{proof}[Proof of Theorem~\ref{thm:multilayered_expand}]
   Fix $\eps,\lambda>0$ and number of layers $\ell\in\mathbb{N}$, and start with a label cover instance $\Psi=((L\cup R,E),\Sigma_L,\Sigma_R,\{\phi_{e}\}_{e\in E})$ from Theorem~\ref{thm:quasilinear_expanding_2} with the constant $C$ as therein. Let $\Psi_{\mathsf{collision}}$ denote the collision game of $\Psi$ as defined in Section~\ref{sec:techniques}. Consider the multilayered label cover instance $\Psi'$ that has the layers $L\times \{i\}$ for $i=1,\ldots,\ell$ with alphabets $\Sigma_L$, and for each $1\leq i<j\leq \ell$, constraints defined by the following process 
\begin{enumerate}
    \item Sample $v\in R$ with probability proportional to its degree.
    \item Sample $u,u'\in L$ neighbours of $v$ independently.
    \item Output the edge $((u,i),(u',j))$ with the constraint 
    \[
    \Phi_{(u,i),(u',j)}=\{(\sigma,\sigma')~|~\exists \tau\in \Sigma_R\text{ such that }\phi_{u,v}(\sigma) = \phi_{u',v}(\sigma')=\tau\}.
    \]
\end{enumerate}

The completeness and soundness are the same as in Section~\ref{sec:222_multi}, and the expansion follows from the expansion of $\Psi$. The game $\Psi'$ is not a projection game though, and to remedy that we move to the game $\Psi''$ defined as follows (a similar ``projection flipping'' operation was done by~\cite{MoshkovitzRaz,DinurHarsha}). 

Denote the left degree and right degree of $\Psi$ by $d_L,d_R$ and let $D = d_L d_R$.
For each $1\leq i\leq \ell$, we replace the alphabet of $(u,i)\in L\times \{i\}$ with the subset of strings from $\Sigma_L^{1+D+D^2+\ldots+D^{\ell-i}}$, which corresponds to an assignment to $u$ and to all vertices $v$ of distance at most $\ell-i$ in $\Psi_{\mathsf{collision}}$. More specifically, interpreting a string as $(\sigma_{v})_{v\in N_{\leq \ell-i}(u)}$, we only allow strings such that for each $(w,w')$ from $\Psi_{\mathsf{collision}}$ among the assigned vertices, the given labels satisfy $\Phi_{w,w'}$.  The edges of $\Psi''$ are the same as those of $\Psi'$, and the constraint on the edge $((u,i),(v,j))$ requires that the labeling of $(u,i)$ and of $(v,j)$ agree on $N_{\leq \ell-j}(v)$.

The completeness and expansion of $\Psi''$ are clear. As per soundness, note that given a labeling to $\Psi''$, we can consider the induced labeling to $\Psi'$ which results by only keeping the label that a vertex $(u,i)$ attributes to itself. If $((u,i),(v,j))$ is satisfied, then the label attributed to $v$ by $(u,i)$ agrees with the label $(v,j)$ attributes $v$, and the local labeling of $(u,i)$ is valid. Hence, these labels satisfy the constraint between $(u,i)$ and $(v,j)$ in $\Psi'$. It follows that $\mathsf{val}(\Psi'')\leq \mathsf{val}(\Psi')$, giving the soundness of the reduction. Finally, the projection follows because in $\Psi_{\mathsf{collision}}$,  the neighbourhood around $v$ of radius $\ell-j$ is contained in the neighbourhood around $u$ of radius $\ell-i$.
\end{proof}

\section{The Hypergraph Vertex Cover Result}
In this section we prove Theorem~\ref{thm:hypergraph_VC_improved_factor}. It is more convenient to prove the result in the language of independent sets. For $0<s<c<1$, an input to the promise problem $\mathsf{gapIS}_k[c,s]$ is a weighted $k$-uniform hypergraph which contains an independent set of weight at least $c$, or no independent set of weight exceeding $s$, and the goal is to distinguish between the two cases.
\begin{thm}\label{thm:hypergraph_IS}
    Let $k\geq 3$ be an integer, and let $\eps>0$ be sufficiently small.
    \begin{enumerate}
        \item The problem $\mathsf{gapIS}_3\left[2-\sqrt{2}-\eps,\eps\right]$ is NP-hard.
        \item For $k\geq 4$,  
        the problem $\mathsf{gapIS}_k\left[1-\frac{1}{k}-\eps,\eps\right]$ is NP-hard.
    \end{enumerate}
\end{thm}
Theorem~\ref{thm:hypergraph_IS} immediately implies Theorem~\ref{thm:hypergraph_VC_improved_factor} using the fact that the complement of an independent set is a vertex cover.

The construction in the proof of Theorem~\ref{thm:hypergraph_IS} combines ideas from~\cite{KhotRegev,DinurSafra} into the construction of~\cite{DGKR}. In~\cite{DGKR} the authors replace each vertex of the label cover instance by a copy of the $p$-biased long code for an appropriately chosen $p$, with the intention of a label being encoded by a dictatorship assignment. They then construct hyperedges that intend to check compatibility of the long-code encodings. More specifically, for each edge $(x,y)$ of the original label cover instance, they add a hyperedge between one vertex $(y,A)$ from the long-code of $y$ and $k-1$ vertices $(x,B_1),\ldots,(x,B_{k-1})$ from the long code of $x$, whenever the sets $A$, $B_1\cap\ldots\cap B_{k-1}$ contain no pair of assignments satisfying the constraint on $(x,y)$. They show a fairly elementary way to analyze this construction for $p=1-\frac{1}{k-1}-\eps$, and thereby gave their $k-1-\eps$ hardness result for vertex cover. 

The earlier work of~\cite{DinurSafra} studies the case that $k=2$, and uses a construction similar in spirit, except that the initial label cover instance can be thought of as a $2$-to-$2$ games instance (this is not quite accurate, but good enough for the sake of discussion). There, the authors manage to show a more involved analysis which works for non-trivial value of $p$. Their analysis amounts to the ability to construct, for each edge $(x,y)$ in the outer label cover instance, a probability distribution over the hyperedges, in which each vertex is $p$-biased distributed in its respective long-code. We refer the reader to~\cite{minzerbook} for a more detailed account.

The key idea in the proof of Theorem~\ref{thm:hypergraph_IS} is to split the vertices of the hyperedge more evenly between the two long codes, so as to make the construction of a coupling as in~\cite{DinurSafra} easier. Intuitively, if $k$ is thought of as a large even number and we wish to take vertices $(x,B_1),\ldots,(x,B_{k/2})$ from one long code and $(y,A_1),\ldots,(y,A_{k/2})$ from the other long code, then the condition that $B_1\cap\ldots\cap B_{k/2}$
and 
$A_1\cap\ldots\cap A_{k/2}$ do not contain a pair of labels satisfying the constraint on $(x,y)$ seems much more mild compared to the condition of~\cite{DGKR}. Indeed, it turns out that once $k\geq 4$ one can construct a suitable  probability distribution for $p=1-\frac{1}{k}-\eps$, and for $k=3$ one can do so up to $p=2-\sqrt{2}-\eps$.
\subsection{The Construction for Theorem~\ref{thm:hypergraph_IS}}

Fix $k\geq 3$ and let $p\in (0,1)$ to be determined.
Take parameters $\ell\in\mathbb{N}$ and $\eps>0$ to be determined, and take $\Psi = \{\Psi_{i,j}\}_{1\leq i<j\leq \ell}$ to be an instance from Theorem~\ref{thm:multilayered_222} with parameters $\ell$ and $\eps$, and denote its layers by $X_1,\ldots,X_{\ell}$ and its alphabet by $\Sigma$. We construct the weighted hypergraph $\mathcal{H} = (\mathcal{V},\mathcal{E})$ as follows:
\begin{enumerate}
    \item The vertex set is $\bigcup_{i=1}^{\ell} \{i\}\times X_i\times P(\Sigma)$, i.e., we replace each vertex $x\in X_i$ with a copy of $P(\Sigma)$ and append a label $i$ for convenience. We define the weight of a vertex $(i,x,A)$ to be $\frac{1}{\ell|X_i|}\mu_p(A)$.
    \item Taking $i<j$, we include in $\mathcal{H}$ the hyperedge $\{(i,x,A_1),(i,x,A_2),(j,y,B_1),\ldots,(j,y,B_{k-2})\}$ if
    \begin{enumerate}
        \item $(x,y)$ is an edge in $\Psi_{i,j}$.
        \item No pair of labels in $A_1 \cap A_2$ and $B_1\cap\ldots\cap B_{k-2}$ satisfy the constraint of $(x,y)$. 
    \end{enumerate}
\end{enumerate}
This completes the description of $\mathcal{H}$, and we next analyze its completeness and soundness.

\subsection{Completeness Analysis}
The completeness of the reduction is straightforward and is proved in the following lemma.
\begin{lemma}\label{lem:completeness}
    If $\Psi$ is as in the Yes case in Theorem~\ref{thm:multilayered_222}, then $\mathcal{H}$ contains an independent set of weight at least $p(1-\eps)$.
\end{lemma}
\begin{proof}
    Let $A_1,\ldots,A_{\ell}$ be partial labelings to $X_1',\ldots,X_{\ell}'$ satisfying all of the constraints between them, and choose 
    \[
    S = \{(i,x,A)~|~i=1,\ldots,\ell, x\in X_i', A_i(x)\in A\}.
    \]
    Then $S$ is an independent set in $\mathcal{H}$, and its weight is $\Expect{i,x\in X_i}{1_{x\in X_i'}\Prob{A\sim\mu_p}{A_i(x)\in A}}
    \geq p(1-\eps)$.
\end{proof}

\subsection{Tools for the Soundness Analysis}
In this section we establish the main new tools that are used in the soundness analysis. For $k\geq 4$, the analysis will use the following construction of a probability distribution:
\begin{lemma}\label{lem:key_coupling_4}
    Suppose $k\geq 4$, $0\leq p\leq 1-\frac{1}{k}$, let $\Sigma_1,\Sigma_2$ be alphabets of the same size, and suppose that $\Phi\subseteq \Sigma_1\times\Sigma_2$ is a $2$-to-$2$ constraint. Then there exists a probability distribution $\mathcal{D}$ over $(A_1,A_2,B_1,\ldots,B_{k-2})$ such that
    \begin{enumerate}
        \item Marginally, each one of $A_1,A_2$ is distributed according to the $p$-biased distribution over $P(\Sigma_1)$, and each one of $B_1,\ldots,B_{k-2}$ is distributed according to the $p$-biased distribution over $P(\Sigma_2)$.
        \item The sets $(A_1\cap A_2)$ and 
        $(B_1\cap\ldots\cap B_{k-2})$ do not contain any pair of labels that satisfy $\Phi$.
    \end{enumerate}
\end{lemma}
\begin{proof}
    By renaming labels we assume that $\Sigma_1=\Sigma_2=\Sigma$.
    We construct the coupling for $p=1-\frac{1}{k}$, and the result follows trivially. Indeed, given the construction for biased parameter $1-\frac{1}{k}$, for $p<1-\frac{1}{k}$ one can first sample from the distribution of $1-\frac{1}{k}$ and then take random subsets by retaining each element independently with probability $p/(1-1/k)$.

    Since $\Phi$ is a $2$-to-$2$ constraint, we may write it as $\bigcup\limits_{i=1}^{m}P_i\times Q_i$ where $\{P_i\},\{Q_i\}$ are partitions of $\Sigma$ into sets of size $2$. We now consider the following randomized process. For each $i=1,\ldots,m$, perform the following process independently:    \begin{enumerate}
        \item With probability $\frac{k-2}{k}$, take $A_{1,i}=A_{2,i} = P_i$ and then
        \begin{enumerate}
            \item With probability $1-\frac{1}{k}$, choose distinct $j_1,j_2\in [k-2]$ uniformly, take $B_{j,i}=Q_i$ for all $j\neq j_1,j_2$, take $B_{j_1,i}, B_{j_2,i}$ to be a uniformly random partition of $Q_i$ into sets of size $1$.
            \item Else, choose $j^{\star}\in [k-2]$ uniformly, set $B_{j^{\star},i}=\emptyset$ and $B_{j,i}=Q_i$ for all $j\neq j^{\star}$.
        \end{enumerate}
        \item With probability $\frac{2}{k}$, take $B_{j,i} = Q_i$ for all $j$ and then
        \begin{enumerate}
            \item With probability $\frac{k-1}{k}$ choose $A_{1,i}, A_{2,i}$ to be a uniformly random partition of $P_i$ into sets of size $1$.
            \item Else, choose $j\in \{1,2\}$ uniformly and take $A_{j,i}=P_i$ and $A_{3-j,i}=\emptyset$.
        \end{enumerate}
    \end{enumerate}
    Then take 
    \[
    A_s=A_{s,1}\cup\ldots \cup A_{s,m}\qquad
    B_r = B_{r,1}\cup\ldots\cup B_{r,m}
    \]
    for $s=1,2$ and $r=1,\ldots,k-2$. By inspecting the process, it is easy to see that $A_{1,i}\cap A_{2,i}$ and $B_{1,i}\cap\ldots\cap B_{k-2,i}$ never contain a pair of labels from $P_i$, $Q_i$ respectively, and thus the second item in the lemma holds. By direct calculation, for each $i$ and $s\in \{1,2\}$,  $A_{s,i}$ is equal to $P_i$ with probability $p^2$, 
    is a random subset of size $1$ with probability $2p(1-p)$, and is equal to $\emptyset$ with probability $(1-p)^2$. The same holds for each $B_{r,i}$ relative to $Q_i$. This implies the first item in the lemma.
\end{proof}

For $k=3$, the construction is a bit different, and is given in the following lemma.
\begin{lemma}\label{lem:key_coupling_3}
    Suppose $k=3$, $0\leq p\leq 2-\sqrt{2}$, and let $\Sigma_1,\Sigma_2$ be alphabets of the same size, and
    $\Phi\subseteq\Sigma_1\times\Sigma_2$ be a $2$-to-$2$ constraint. Then there exists a
    probability distribution $\mathcal{D}$ over $(A_1,A_2,B)$ such that
    \begin{enumerate}
        \item Marginally, each one of $A_1,A_2$ is distributed
        according to the $p$-biased distribution over $P(\Sigma_1)$, and $B$ is distributed according to the $p$-biased distribution over $P(\Sigma_2)$.
        \item The sets $A_1\cap A_2$ and $B$ do not contain any pair of labels satisfying $\Phi$.
    \end{enumerate}
\end{lemma}
\begin{proof}
    As before we assume that $\Sigma_1=\Sigma_2=\Sigma$, and it suffices to prove the lemma for $p=2-\sqrt{2}$. We use the same $P_i$ and $Q_i$ notations as before, and for each $i=1,\ldots,m$ independently perform the following randomized process:
    \begin{enumerate}
        \item With probability $(1-p)^2$, take    $A_{1,i}=A_{2,i}=P_i$ and $B_{i}=\emptyset$.

        \item With probability $2(1-p)$, perform the following two choices
        independently.
        \begin{enumerate}
            \item With probability $p$, take $A_{1,i},A_{2,i}$
            to be a uniformly random partition of $P_i$ into sets of
            size $1$, and else choose $s\in\{1,2\}$
            uniformly and take
            $A_{s,i}=P_{i}$, $A_{3-s,i}=\emptyset$.

            \item With probability $p$ take $B_{i}$ to be a
            uniformly random subset of $Q_i$ of size $1$, and else take it to be $Q_i$.
        \end{enumerate}
    \end{enumerate}
    This process is possible to execute because $(1-p)^2+2(1-p) = (1-p)(3-p)=(\sqrt{2}-1)(1+\sqrt{2})=1$.
    We then take $A_s = A_{s,1}\cup\ldots\cup A_{s,m}$ for $s=1,2$ and 
    $B = B_{1}\cup\ldots\cup B_{m}$. The two items in the lemma follow from a direct inspection of the process.
\end{proof}

\subsection{Soundness Analysis}
We now analyze the soundness of the reduction. Define $p_c = 1-\frac{1}{k}$ if $k\geq 4$ and $p_c = 2-\sqrt{2}$ if $k=3$. In this section we prove the soundness of the reduction, which amounts to the following lemma.
\begin{lemma}\label{lem:soundness_hypergraph}
    For all $\eta\in (0,\frac{1}{100 k})$, for sufficiently large $\ell$, sufficiently small $\lambda>0$ and sufficiently small $\eps>0$, if $p= p_c-\eta$, $\Psi$ is $(1/2,\lambda)$-weakly dense and $\mathsf{val}(\Psi_{i,j})\leq \eps$ for all $i<j$, then the largest independent set in $\mathcal{H}$ has weight at most $\eta$.
\end{lemma}
The rest of this section is devoted to the proof of Lemma~\ref{lem:soundness_hypergraph}, and we assume towards contradiction that $\mathcal{I}\subseteq \mathcal{V}$ is an independent set of weight at least $\eta$. We pick $\ell = \frac{128}{\eta^2}$ and $\lambda = \frac{\eta}{8}$.
Define the up-closure of $\mathcal{I}$ to be
\[
\mathcal{I}^{\uparrow} = \sett{(i,x,B)}{\exists A\subseteq B, (i,x,A)\in \mathcal{I}}.
\]
It is easy to see that $\mathcal{I}^{\uparrow}$ is also an independent set and its weight is at least as large as that of $\mathcal{I}$. Thus, we may work with $\mathcal{I}^{\uparrow}$ instead of $\mathcal{I}$, and for notational convenience we assume that $\mathcal{I} = \mathcal{I}^{\uparrow}$ to begin with.

For each $i=1,\ldots,\ell$ and $x\in X_i$, define
\[
\mathcal{I}_{i,x} = \sett{A\subseteq \Sigma}{(i,x,A)\in \mathcal{I}}.
\]
Then 
\[
\Expect{i,x\in X_i}{\mu_p(\mathcal{I}_{i,x})}
\geq \eta,
\]
so by an averaging argument for at least $\frac{\eta}{2}$ fraction of $(i,x)$ it holds that $\mu_p(\mathcal{I}_{i,x})\geq\frac{\eta}{2}$, and we refer to such $(i,x)$ as good. Next, note that by the mean value theorem, there exists $p'\in (p,p+\eta)$ such that
\[
\Expect{i,x\in X_i}{\frac{d\mu_z(\mathcal{I}_{i,x})}{dz}(p')}
=
\frac{d\Expect{i,x\in X_i}{\mu_z(\mathcal{I}_{i,x})}}{dz}(p')=
\frac{\Expect{i,x\in X_i}{\mu_{p+\eta}(\mathcal{I}_{i,x})}
-\Expect{i,x\in X_i}{\mu_p(\mathcal{I}_{i,x})}}{\eta}\leq \frac{1}{\eta}.
\]
We fix this $p'$. Note that since $\mathcal{I}$ is upward closed, $\mathcal{I}_{i,x}$ is monotone for all $i, x$, so by Lemma~\ref{lem:russo_margulis} we conclude that 
$\Expect{i,x\in X_i}{I[\mathcal{I}_{i,x};\mu_{p'}]}\leq \frac{1}{\eta}$. Markov's inequality now gives that for at least $1-\frac{\eta}{4}$ fraction of $(i,x)$ we have that 
$I[\mathcal{I}_{i,x};\mu_{p'}]\leq \frac{4}{\eta^2}$. 
Using Lemma~\ref{lem:russo_margulis} again, it follows that 
$z\rightarrow \mu_z(\mathcal{I}_{i,x})$ is monotone, and hence
$\mu_{p'}(\mathcal{I}_{i,x})\geq \mu_p(\mathcal{I}_{i,x})
$.
Overall, we get that for at least $\frac{\eta}{4}$ fraction of $(i,x)$ we have
\begin{equation}\label{eq:prop_of_subfam}
\mu_{p'}(\mathcal{I}_{i,x})\geq \frac{\eta}{2},
\qquad 
I[\mathcal{I}_{i,x};\mu_{p'}]\leq \frac{4}{\eta^2},
\end{equation}
and we refer to such $(i,x)$ as great.

Applying Lemma~\ref{lem:friedgut}, for each great $(i,x)$ we may find $J_{i,x}\subseteq \Sigma$ of size at most $C=C(\eta)$ and a $J_{i,x}$ junta $\mathcal{J}_{i,x}\subseteq P(\Sigma)$ such that 
\begin{equation}\label{eq:prop_of_fried}
\mu_{p'}(\mathcal{I}_{i,x}\Delta \mathcal{J}_{i,x})\leq \frac{\eta^2}{8}.
\end{equation}
\begin{definition}
    We define the extended junta of $(i,x)$ by
    \[
    EJ_{i,x} = J_{i,x}\cup\sett{\sigma\in\Sigma}{I_{\sigma}[\mathcal{I}_{i,x};\mu_{p'}]\geq 2^{-10k C}\eta}.
    \]
\end{definition}
We note that $|EJ_{i,x}|\leq C+ \frac{2^{10kC}}{\eta} I[\mathcal{I}_{i,x};\mu_{p'}]\leq \frac{8\cdot 2^{10k C}}{\eta^3}$.

\begin{claim}\label{claim:extended_juntas_intersect}
    Suppose that $(x,y)$ is an edge in $\Psi_{i,j}$, and that $(i,x)$ and $(j,y)$ are great. Then 
    $EJ_{i,x}$ and $EJ_{j,y}$
    contain a pair of labels that satisfy the constraint $\Phi_{x,y}$.
\end{claim}
\begin{proof}
    Assume for contradiction this is not the case. It will be convenient to think of the $2$-to-$2$ constraint as defined by a pair of $2$-to-$1$ maps $\pi_1\colon \Sigma\to \Gamma$, 
    $\pi_2\colon \Sigma\to\Gamma$, where $|\Gamma|=|\Sigma|/2$ so that
    \[
    \Phi_{x,y}=\sett{(\sigma,\sigma')}{\pi_1(\sigma)=\pi_2(\sigma')}.
    \]
    Thus, the assumption towards contradiction can be written as $\pi_1(EJ_{i,x})\cap \pi_2(EJ_{j,y}) = \emptyset$. We will show that this implies that $\mathcal{I}$ is not an independent set, and for that we will construct $A_1,A_2\in \mathcal{I}_{i,x}$ and 
    $B_1,\ldots,B_{k-2}\in\mathcal{I}_{j,y}$ that together form a hyperedge.

    First, combining~\eqref{eq:prop_of_subfam} and~\eqref{eq:prop_of_fried} we conclude that $\mu_{p'}(\mathcal{J}_{i,x})\geq \frac{\eta}{4}$, so sampling $A'\subseteq J_{i,x}$ from the $p'$-biased measure, we have that $\mu_{p'}((\mathcal{J}_{i,x})_{J_{i,x}\rightarrow A'})=1$ with probability at least $\frac{\eta}{4}$. By~\eqref{eq:prop_of_fried} and Markov's inequality, we get that 
    $\mu_{p'}((\mathcal{I}_{i,x}\Delta \mathcal{J}_{i,x})_{J_{i,x}\rightarrow A'})\leq \eta$ with probability at least $1-\frac{\eta}{8}$, and together we get that with probability at least $\frac{\eta}{8}$ it holds that 
    \[
    \mu_{p'}((\mathcal{I}_{i,x})_{J_{i,x}\rightarrow A'})\geq 1-\eta.
    \]
    Defining $\mathsf{comp}(J_{i,x}) = \pi_1^{-1}(\pi_1(J_{i,x}))$ and 
    $\mathsf{comp}(J_{j,y}) = \pi_2^{-1}(\pi_2(J_{j,y}))$, it follows we may find $A'\subseteq \mathsf{comp}(J_{i,x})$ such that 
    $\mu_{p'}((\mathcal{I}_{i,x})_{\mathsf{comp}(J_{i,x})\rightarrow A'})\geq 1-\eta$. As $J_{j,y}\subseteq EJ_{j,y}$, the premise implies that $\pi_{1}^{-1}(\pi_2(J_{j,y}))\cap EJ_{i,x}=\emptyset$, so elements in $\pi_1^{-1}(\pi_2(J_{j,y}))$ have at most $2^{-10k C}\eta$ influence on $\mathcal{I}_{i,x}$.
    Thus, they have at most 
    \[
    k^{\card{\mathsf{comp}(J_{i,x})}}2^{-10k C}\eta\leq k^{2C}2^{-10k C}\eta
    \leq 2^{-7kC}\eta
    \] 
    influence on $(\mathcal{I}_{i,x})_{\mathsf{comp}(J_{i,x})\rightarrow A'}$. Conditioning on any additional coordinate may increase the individual influences by a factor of $k$ at most. It follows that
    \begin{align}\label{eq:measure_of_x}
    \mu_{p'}((\mathcal{I}_{i,x})_{\mathsf{comp}(J_{i,x})\cup \pi_{1}^{-1}(\pi_2(J_{j,y})) \rightarrow A'})
    &\geq 1-\eta-|\pi_{1}^{-1}(\pi_2(J_{j,y}))|k^{|\pi_{1}^{-1}(\pi_2(J_{j,y}))|}2^{-7k C}\eta \notag\\
    &\geq 
    1-\eta-2 C\cdot 2^{-4kC}\eta \notag\\
    &\geq 1-2\eta.
    \end{align}
    Similarly, for $(j,y)$ we may find $B'\subseteq \mathsf{comp}(J_{j,y})$ such that
    \begin{equation}\label{eq:measure_of_y}
    \mu_{p'}((\mathcal{I}_{j,y})_{\mathsf{comp}(J_{j,y})\cup \pi_{2}^{-1}(\pi_1(J_{i,x})) \rightarrow B'})
    \geq 1-2\eta.
    \end{equation}  
    Define $D_x = \mathsf{comp}(J_{i,x})\cup \pi_{1}^{-1}(\pi_2(J_{j,y}))$ and 
    $D_y = \mathsf{comp}(J_{j,y})\cup \pi_{2}^{-1}(\pi_1(J_{i,x}))$. Note that the restriction of $\pi_1$ from $\Sigma\setminus D_x$ to $\Gamma\setminus \pi_1(D_x)$, denoted by $\pi_1'$, is a $2$-to-$1$ map, that the restriction of $\pi_2$ from $\Sigma\setminus D_y$ to $\Gamma\setminus \pi_2(D_y)$, denoted by $\pi_2'$, is a $2$-to-$1$ map, and $\pi_1(D_x) = \pi_2(D_y)$. Thus, together $\pi_1'$ and $\pi_2'$ define a $2$-to-$2$ constraint between $\Sigma\setminus D_x$ and $\Sigma\setminus D_y$. Applying Lemmas~\ref{lem:key_coupling_4} or~\ref{lem:key_coupling_3} with parameter $p'$ we may find a probability distribution $(A_1,A_2,B_1,\ldots,B_{k-2})$ as therein. By the marginals property and~\eqref{eq:measure_of_x},~\eqref{eq:measure_of_y} we get that
    \[
    A'\cup A_1, A'\cup A_2\in \mathcal{I}_{i,x},
    \qquad
    B'\cup B_1,\ldots,B'\cup B_{k-2}\in\mathcal{I}_{j,y}
    \]
    with probability at least $1-2k\eta>0$, and we fix such $(A_1,A_2,B_1,\ldots,B_{k-2})$. By Lemmas~\ref{lem:key_coupling_4} or~\ref{lem:key_coupling_3} the sets $A_1\cap A_2$ and $B_1\cap\ldots\cap B_{k-2}$ do not contain a pair of labels satisfying the $2$-to-$2$ constraint between $\Sigma\setminus D_x$ and $\Sigma\setminus D_y$. By construction, $A'$ and $B'$ do not contain a pair of labels satisfying the $2$-to-$2$ constraint between $D_x$ and $D_y$. It follows that there is a hyperedge in $\mathcal{H}$ between the vertices 
    \[
    (i,x,A'\cup A_1),
    (i,x,A'\cup A_2),
    (j,y,B'\cup B_1),
    \ldots,
    (j,y,B'\cup B_{k-2}),
    \]
    and contradiction.
\end{proof}
We now finish the soundness analysis. Since at least $\frac{\eta}{4}$ fraction of $(i,x)$ are great, it follows that for at least $\frac{\eta}{8}$ fraction of layers $i$, we have that the fraction of $x$ such that $(i,x)$ is great is at least $\frac{\eta}{8}$. Let $I$ be the set of these layers, and for each $i\in I$ let
\[
X_i'=\sett{x\in X_i}{(i,x)\text{ is great}}.
\]
By the weak density, we may find $i<j$ in $I$ such that at least $(\eta/16)^{2}$ fraction of the edges in $\Psi_{i,j}$ are between $X_{i}'$ and $X_{j}'$. For each $x\in X_i'$, choose its label $A(x)$ uniformly from $EJ_{i,x}$, and similarly for each $y\in X_j'$. Then by Claim~\ref{claim:extended_juntas_intersect}, it follows that the probability that the chosen labels satisfy the constraint $\Phi_{x,y}$ is at least
\[
\frac{1}{\card{EJ_{i,x}}\card{EJ_{j,y}}}
\geq \left(\frac{\eta^3}{8\cdot 2^{10kC}}\right)^{2},
\]
so the expected fraction of constraints in $\Psi_{i,j}$ that are satisfied is at least
\[
\left(\frac{\eta}{16}\right)^2\left(\frac{\eta^3}{8\cdot 2^{10kC}}\right)^{2}=\Omega_{\eta}(1) > \eps
\]
for sufficiently small $\eps>0$. This completes the proof of Lemma~\ref{lem:soundness_hypergraph}.

\bibliographystyle{alpha}
\bibliography{references}
\appendix
\end{document}